\documentclass[copyright,creativecommons]{eptcs}
\providecommand{\event}{LFMTP'23} 

\usepackage{iftex}

\ifpdf
  \usepackage{underscore}         
  \usepackage[T1]{fontenc}        
\else
  \usepackage{breakurl}           
\fi
\AtBeginDocument{}

\usepackage{wrapfig}
\usepackage{booktabs}
\usepackage{proof}
\usepackage{amsmath}
\usepackage{amsthm}
\usepackage{amssymb}
\usepackage{mathrsfs}
\usepackage{comment}    
\usepackage{algorithm}
\usepackage[noend]{algpseudocode}

\usepackage{tabularx}

\usepackage{todonotes}

\newtheorem{theorem}{Theorem}

\newcommand{\ignore}[1]{}

\newcommand{\ND}[0]{$\mathcal{N\mkern-4mu D}$ }

\title{Parallel Verification of Natural Deduction Proof Graphs}
\author{
    James T. Oswald
    \qquad\qquad Brandon Rozek
    \institute{
        Rensselaer \textit{AI} \& Reasoning (RAIR) Lab, \\
        Rensselaer Polytechnic Institute (RPI) \\
        Troy, New York, USA
    }
    \email{\{oswalj, rozekb\}@rpi.edu}
}
\newcommand{\titlerunning}{Parallel Verification of Natural Deduction Proof Graphs}
\newcommand{\authorrunning}{J. Oswald and B. Rozek}

\hypersetup{
  bookmarksnumbered,
  pdftitle    = {\titlerunning},
  pdfauthor   = {\authorrunning},
  pdfkeywords = {Natural Deduction, Formal Verification, Parallel Computing}
}

\begin{document}

\maketitle

\begin{abstract}
Graph-based interactive theorem provers offer a visual representation of proofs, explicitly representing the dependencies and inferences between each of the proof steps in a graph or hypergraph format. The number and complexity of these dependency links can determine how long it takes to verify the validity of the entire proof. Towards this end, we present a set of parallel algorithms for the formal verification of graph-based natural deduction (\ND) style proofs. We introduce a definition of layering that captures dependencies between the proof steps (nodes). Nodes in each layer can then be verified in parallel as long as prior layers have been verified. To evaluate the performance of our algorithms on proof graphs, we propose a framework for finding the performance bounds and patterns using directed acyclic network topologies (DANTs). This framework allows us to create concrete instances of DANTs for empirical evaluation of our algorithms. With this, we compare our set of parallel algorithms against a serial implementation with two experiments: one scaling the problem size and the other scaling the number of threads. Our findings show that parallelization results in improved verification performance for certain DANT instances. We also show that our algorithms scale for certain DANT instances with respect to the number of threads.
\end{abstract}

\section{Introduction}

A major role of an interactive theorem prover is to take
an existing proof and verify that it is valid with respect
to the logical calculi used. This involves iterating
over each step of a proof and verifying both that it 
syntactically matches the transformation of the formula under
the rule, and that it is valid with respect to the semantics
of the underlying proof system.
Interactive theorem provers such as Coq \cite{thecoqdevelopmentteam20193476303},
Lean \cite{moura2021lean}, and HyperSlate \cite{Bringsjord2022}
verify not only that the proof written by the user is correct,
but also every underlying proof that the given proof depends on.
This generally amounts to verifying large portions of the standard library
and other popular libraries such as mathlib \cite{Community2019TheLM}.
Our work makes a step toward speeding up the proof
verification process.
We focus on the verification of natural deduction proof graphs,
such as those represented in HyperSlate, though the ideas
from this approach could be adapted to other interactive theorem provers as well.

In order to speed up verification, we look toward parallel computing.
One naive implementation would be to verify all the proof steps in
parallel. This assumes, however, that the step has all the semantic
information needed to show validity. This is often not the case
for many logic calculi. Assumptions are introduced and discharged
in the case of natural deduction. Variables may be assigned
to constants. These issues present a constraint that in
order to parallelize verification, we need to ensure that
some steps are verified before others.
We achieve this by introducing a layering approach.
Given a definition of layering that induces a topological partial order,
every step within a layer $n$ only depends on steps within
the layers prior. 
Given these layers, we can then verify all the steps that are from
the same layer in parallel without worrying about invalidating the
underlying semantics.
To illustrate this approach, we present the parallel
verification of natural deduction proof graphs.

The underlying dependency nature of each of the steps induces
a directed acyclic hypergraphical representation where nodes hold \ND statements and hyperedges between nodes represent inference rules. 
This graphical representation not only gives us an easy
way to visualize such proofs, but also provides insight
on empirically validating our parallel algorithms. 
Inspired by computer network topologies, we introduce
directed acyclic network topologies (DANTs) as a way to
identify classes of graphical proofs.
These topologies provide a method of comparing the
performance of
different verification strategies on various proof structures.

The contributions of this work are as follows: (1) A layering
approach that decouples the dependencies of proof steps
within \ND proofs.
(2) Parallel verification algorithms that outperform
serial verification on non-straight topologies and scales 
with the number of hardware-based threads.
(3) Introduction of several classes of
graphical proofs, with an eye on empirical evaluation.
The relevant background which includes \ND and hypergraphical representations
is discussed in \S \ref{sec:Background}.
Within \S \ref{sec:Approach}, we discuss the
proof verification algorithm and several
optimizations. Then in \S \ref{sec:Results}, we discuss
directed acyclic network topologies to empirically
evaluate common proof structures.
In that section, we also discuss our performance and scaling results.
We then conclude by talking about
related work in \S \ref{sec:Related-Work}.

\section{Background}\label{sec:Background}
\subsection{Natural Deduction}

Natural deduction (\ND) is a logic calculus independently
proposed in \cite{gentzen1935untersuchungen, jaskowski1934rules}
in an effort to emulate human-level reasoning through
assumptions and chains of inference. There are many different styles of proof that fall under natural deduction, the three most common come from Gentzen \cite{gentzen1935untersuchungen}, Jaśkowski \cite{jaskowski1934rules}, and Fitch \cite{fitch1952symbolic}. However, we are mainly interested in a style that interoperates with a hypergraphical representation of 
natural deduction proofs. 

\begin{figure}[h!]
\begin{gather*}
\infer[\text{A}]{\{\phi\} \vdash \phi}{}
\quad\quad \infer[\land I]{\Gamma \cup \Sigma \vdash \phi \land \psi}{\Gamma \vdash \psi \quad \Sigma \vdash \phi}
\quad\quad \infer[\land E_l] {\Gamma \vdash \phi}{\Gamma \vdash \phi \land \psi} 
\quad\quad \infer[\land E_r] {\Gamma \vdash \psi}{\Gamma \vdash \phi \land \psi} \\
\quad\quad \infer[\lor I_l] {\Gamma \vdash \psi \lor \phi }{\Gamma \vdash \phi} 
\quad\quad \infer[\lor I_r] {\Gamma \vdash \phi \lor \psi}{\Gamma \vdash \phi} 
\quad\quad \infer[\lor E] {\Delta \cup \Gamma \cup \Sigma  \vdash \chi}{\Delta \vdash \psi \lor \phi \quad \Gamma \cup \{\psi\} \vdash \chi \quad \Sigma \cup \{\phi\} \vdash \chi} \\
\quad\quad \infer[\rightarrow I]{\Gamma \vdash \phi \rightarrow \psi}{\Gamma \cup \{\phi\} \vdash \psi}
\quad\quad \infer[\rightarrow E]{\Gamma \cup \Sigma \vdash \psi}{\Gamma \vdash \phi \quad \Sigma \vdash \phi \rightarrow \psi} \\
\quad\quad \infer[\lnot I]{\Gamma \cup \Sigma \vdash \lnot \phi}{\Gamma \cup \{\phi\} \vdash \psi \quad \Sigma \vdash \lnot \psi}
\quad\quad \infer[\lnot E]{\Gamma \cup \Sigma \vdash \phi}{\Gamma \cup \{\lnot \phi\} \vdash \psi \quad \Sigma \vdash \lnot \psi} \\
\quad\quad \infer[\leftrightarrow I]{\Gamma \cup \Sigma \vdash \phi \leftrightarrow \psi}{\Gamma \cup \{\phi\} \vdash \psi \quad \Sigma \cup \{\psi\} \vdash \phi} 
\quad\quad \infer[\leftrightarrow E_l]{\Gamma \cup \Sigma \vdash \psi}{\Gamma \vdash \phi \quad \Sigma \vdash \phi \leftrightarrow \psi} 
\quad\quad \infer[\leftrightarrow E_r]{\Gamma \cup \Sigma \vdash \phi}{\Gamma \vdash \psi \quad \Sigma \vdash \phi \leftrightarrow \psi}
\end{gather*}
\caption{Our inference schemata for natural deduction. Within each schema, $\Gamma, \Sigma, \Delta$ are sets of formulae, and $\phi, \psi, \chi$ are meta-logical variables which range over formulae. Note that our formulation of $\lnot I, \lnot E, \leftrightarrow I, \leftrightarrow E$ differs from those typically seen in other works such as \cite{Prawitz1965NaturalDA} but are equivalent.}
\label{fig:rules}
\end{figure}

\newpage

In this paper, we focus on propositional natural deduction.
Let $p$ denote an atomic proposition. The language of propositional logic may be
defined inductively using Backus Naur Form (BNF) as the following:
\begin{equation*}
    \phi ::= p | \neg \phi | (\phi \wedge \phi) | (\phi \vee \phi) | (\phi \rightarrow \phi) | (\phi \leftrightarrow \phi)
\end{equation*}
Our inference rules for \ND are summarized in Figure \ref{fig:rules} \footnote{While on the surface this formalization may appear similar to sequent natural deduction\cite{Pelletier2023}, we use ``$\vdash$'' in this formalism to mean syntactic entailment, with $\Gamma \vdash \phi$ being read as "Assuming $\Gamma$, then $\phi$" or "$\phi$ can be derived from $\Gamma$".}. This formalization is modeled after Bringsjord \cite{bringsjordintermediate}
and fully captures the notion of discharging of assumptions. It is also particularly well suited
to hypergraphical representation, which will be discussed in \S \ref{sec:Hypergraph}. 
These inference rules can be broadly split into two categories: (1) introduction rules ($\land I, \lor I_l, \lor I_r, \lnot I, \rightarrow I, \leftrightarrow I$), in which a logical connective is introduced into the conclusion, and (2) elimination rules ($\land E_l, \land E_r, \lor E, \lnot E, \rightarrow E, \leftrightarrow E_l, \leftrightarrow E_r$), in which a connective in a rule's premise is removed in its conclusion. The outlier here is the Assumption rule (A) which allows us to assert $\{\phi\} \vdash \phi$, or in English, "assuming $\phi$, $\phi$ follows".

For a natural deduction proof, a step is considered valid if the formula
is well-formed and it is justified by a rule of inference.
Valid formulae with
no assumptions are called tautologies. A proof is considered valid iff all of its steps are valid. An example of a valid proof can be seen in Figure \ref{fig:exampleProof1}.

\begin{figure}[!h]
\begin{equation*}
\infer[\lor E]{\{A \lor B, \lnot A\} \vdash B}{
\infer[\text{A}]{\{A \lor B\} \vdash A \lor B}{} \quad
\infer[\lnot E]{\{\lnot A, A\} \vdash B}{
\infer[\land E_l]{\{\lnot B, \lnot A\} \vdash \lnot A}{
\infer[\land I]{\{\lnot B, \lnot A\} \vdash \lnot A \land \lnot B}{
\infer[\text{A}]{\{\lnot A\} \vdash \lnot A}{} \quad
\infer[\text{A}]{\{\lnot B\} \vdash \lnot B}{}
}
} \quad
\infer[\text{A}]{\{A\} \vdash A}{}
} \quad
\infer[\text{A}]{\{B\} \vdash B}{}
}
\end{equation*}
\caption{An example of a valid proof of $B$ from $\{A \vee B, \neg A\}$. All steps are valid since at each step (1) all formulae are well formed and (2) the provided rule of inference can be legally applied at each stage given the current assumptions and premises.}
\label{fig:exampleProof1}
\end{figure}

\subsection{Hypergraphical Representation}\label{sec:Hypergraph}

\begin{figure}[ht]
    \centering
    \includegraphics[width=0.8\textwidth]{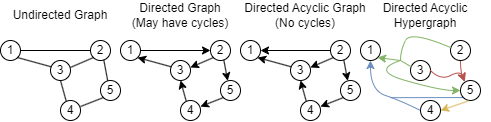}
    \caption{Visualizations of 4 different types of graphs, note that in the hypergraph, edges that share the same color are the same edge.}
    \label{fig:hypergraphs}
\end{figure}

A natural deduction proof can be represented diagrammatically
as a directed acyclic hypergraph
\cite{Voloshin2013IntroductionTG,arkoudas2009vivid}.
A directed acyclic hypergraph is a generalization of
a directed acyclic graph (DAG) which is a mathematical
structure $(V, E)$ where $V$ is a set of vertices and
$E: (V \times V)$ is a set of pairs of vertices.
Acyclic in this context means that for any vertex
$v$, it is not possible to find a path following the
directed edges that leads to $v$.
Directed acyclic hypergraphs extend this by allowing a set of vertices to be connected to a set of vertices by a single edge, thus a directed acyclic hypergraph is a structure $(V, E)$ where $V$ is a set of vertices and $E: \mathcal{P}(V) \times \mathcal{P}(V)$, where $\mathcal{P}(V)$ is the power-set of the set of vertices. Figure \ref{fig:hypergraphs} shows visualizations of the three graph formalisms described. 

To represent natural deduction proofs as hypergraphs, vertices represent premises and conclusions,
and edges represent inference rules. A \textit{proof graph} will be defined as a hypergraph of the form $(V, E)$ where $V$ is a set of statements in the form $\Gamma \vdash \phi$ and $E : \mathcal{P}(V) \times V$ is the set of directed hypergraphical edges representing inference rules applied between statements in the proof.\footnote{Note we use $\mathcal{P}(V) \times V$ rather than $\mathcal{P}(V) \times \mathcal{P}(V)$ since for all inference rules enumerated in Figure \ref{fig:rules} there is only one conclusion, thus each hyper-edge representing an inference rule will only ever have one outgoing connection.} This formalism underlies the representation of proofs in graphical interactive theorem provers such as \cite{Bringsjord2022,Oswald2022}. Figure \ref{fig:NDtoHypergraphs} provides examples of two natural deduction proofs that have been converted to hypergraphical form. 

Interactive theorem provers often do not force the user to keep track of proof state.
Therefore, it is important to note that we are interested in verifying proof graphs
where the assumptions on each node are yet to be known.
We are only given the $\phi$ on each node and must compute the $\Gamma$ based on how
the assumptions update within the inference rules.
If we had both $\Gamma$ and $\phi$, parallelization would be trivial, since we can then verify all the nodes in parallel.

\begin{figure}[ht]
\begin{minipage}{0.59\textwidth}
\begin{equation*}
\infer[\lnot I]{\{\lnot (p \lor q)\} \vdash \lnot p}{
\infer[\lor I_r]{\{p\} \vdash p \lor q}{
\infer[\text{A}]{\{p\} \vdash p}{}
} \quad
\infer[\text{A}]{\{\lnot (p \lor q)\} \vdash \lnot (p \lor q)}{}
}
\end{equation*}
\centering
\includegraphics[width=0.59\textwidth]{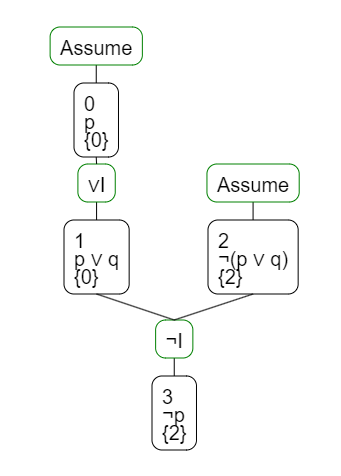}
\end{minipage}
\begin{minipage}{0.39\textwidth}
    \includegraphics[width=\textwidth]{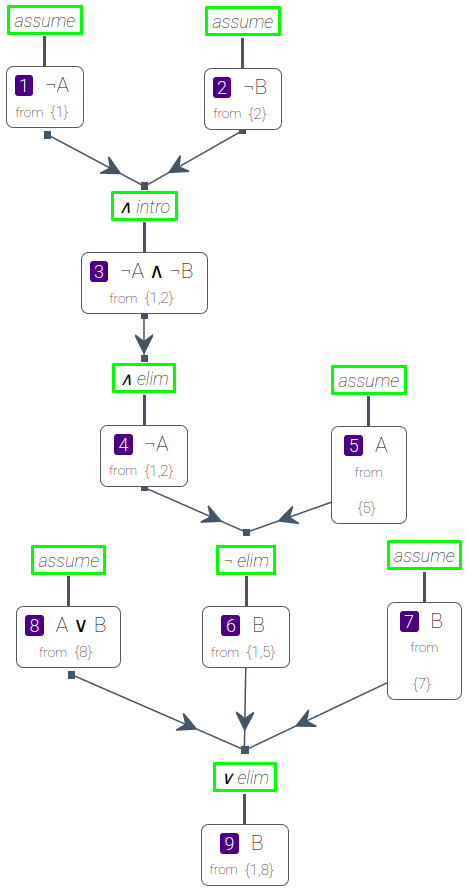}
\end{minipage}
\caption{(Top Left) A valid natural deduction proof that $\{\lnot (p \lor q)\} \vdash \lnot p$ (Bottom Left) The same proof represented as a hypergraphical proof graph structure in the Lazyslate interactive theorem prover\cite{Oswald2022}. (Right) Proof graph of the proof from Figure \ref{fig:exampleProof1} of $\{A \vee B, \neg A\} \vdash B$ in the HyperSlate \cite{Bringsjord2022} interactive theorem prover.}
\label{fig:NDtoHypergraphs}
\end{figure}

Proof graphs implicitly provide additional useful features for representing collections of natural deduction proofs, particularly those that are commonly added onto natural deduction via additional formalisms. First, proof graphs provide the ability to compactly represent proofs that contain reoccurring subproofs. This is because each hyper-node may have multiple outgoing hyper-edges, representing multiple inferences it is used in. While this feature is implicitly captured by proof graphs, natural deduction proofs require an additional formalism allowing named theorems that can be used in other proofs to provide this functionally. Another feature is that a single proof graph can have multiple conclusions or even contain multiple proofs where each proof is a disjoint hyper-subgraph. Without proof graphs, the ability to represent this feature would require a formalism in which a set of proofs can be be treated as a single proof.

\subsection{Multiprocessing}
In this work we use a shared memory model for
multiprocessing. This involves multiple threads
independently operating over the same shared memory space.
More specifically, we make heavy use of
single program multiple data (SPMD) style programs.
With a fixed number of threads instantiated, we attempt to distribute
work evenly across all the threads. 
Our work makes use of two important concepts from multi-processing: thread-safety and reductions (see \cite{Kumar1994} for more extensive coverage). \textit{Thread-safety} within a shared memory model is the notion that parallel algorithms are safe from errors due to concurrent writes and reads from the same piece of memory, known as a \textit{race condition}. 
\ignore{We have made all of our algorithms thread-safe.
}
The second concept is the notion of a parallel reduction. A \textit{parallel reduction} is an operator that takes a list of elements and computes a single element in parallel. A small example of this would be parallel sum over the list $(1,2,3,4)$: if we add 1 and 2 on one thread and 3 and 4 on another thread in parallel, and then sum their results, we can sum the entire list in 2 steps rather than the 3 steps it would take to sum the list in serial.   


\section{Approach}\label{sec:Approach}

We mentioned in \S \ref{sec:Background} that
natural deduction makes use of assumptions and chains
of inference in its proofs.
In the hypergraphical representation, to show that
a given node is valid, we need to show both that the
\emph{syntactic transformation} is valid and
that the \emph{assumption constraints} are met with
respect to the justification used inside that proof.
Let us consider the rule disjunction elimination
(more commonly known as proof by cases) from
Figure \ref{fig:rules} and its usage in the right proof
of Figure \ref{fig:NDtoHypergraphs}.
For example, we wish to show that
bottom node $B$ is valid.
For the syntactic check, we need to ensure that there
are three parent nodes, one of them is a disjunct, and
two of the other parent nodes match the current node.
Then for the assumption constraints, we need to make sure
that for one parent node $B$ it has $A$ in its assumption
set, and for the other parent node $B$ it has $B$ in its
assumption set.
Note from our discussion of the hypergraphical
representation in \S \ref{sec:Hypergraph} that the
underlying nodes do not contain the assumption information
themselves, but they are computed by the application of
each inference rule.  This creates the need of an additional
data structure that we call \texttt{assumptions} during
the verification process.
We obtain the justification
of a given step by calling \texttt{just} on
the node. This will return the justification that
is stored on the incoming edge of the node.

For a baseline comparison with the parallel verification algorithms,
we designed a single-threaded implementation that shares the same
algorithmic structure as the parallel ones minus the usage of shared
memory and threading.
The benchmark results are further discussed in \S \ref{sec:Results}.
Our algorithm works by maintaining a global map
of nodes to their set of assumptions. To ensure that
a node does not get verified before its parent, we
make use of a layering approach which induces a \textit{topological partial ordering} on the nodes. A topological partial ordering, also referred to as \textit{topological generations}, of a proof graph $G = (V, E)$ is a partial ordering $\preccurlyeq$ on the nodes $V$ where for each hyperedge $(\{v_{i0}, \cdots, v_{in}\}, v_o)$, all incoming nodes $\{v_{i0}, \cdots, v_{in}\}$ appear before the outgoing node $v_o$, that is $\forall (V_i, v_o) \in E : \forall v_{ij} \in V_i : v_{ij} \preccurlyeq v_o$.  
This layering approach for generating a topological partial ordering is similar to the well known serial topological sort algorithm \cite{Cormen2009} which generates a topological linear ordering of the nodes but lacks parallelizability.   
Figure \ref{fig:example} provides a colored example of
the nodes on each layer. More formally,
we define node $n$ to be on a layer $L(n)$ inductively as follows:

\begin{equation}
    \label{eqn:layer}
    L(n) =
        \begin{cases}
            0 & \text{if $n$ is an assumption} \\
            1 + \max_{m \in P(n)}(L(m)), & \text{otherwise}
        \end{cases}
\end{equation}
where $P(n)$ maps a node to its parents.

\subsection{Single-Threaded Implementation}

\begin{algorithm}[ht]
\caption{Single-Threaded Algorithm}
\label{alg:single}
\begin{algorithmic}[1]
\Procedure{verify}{ProofGraph p}
\State Initialize \texttt{assumptions} to be empty.
\State Create set of nodes on each layer using Equation \ref{eqn:layer} and store in layerMap.

\For {layerNodes in layerMap}
    \For {n in layerNodes}
        \State justification = \texttt{just}(n)
        \State ruleInfo = (m, \texttt{assumptions}(m)) $\forall$ m $\in$ parents(n)
        \If {not is\_valid(n, justification, ruleInfo)}
            \State return false
        \EndIf
        \State Update \texttt{assumptions}(n) using the justification and ruleInfo.
    \EndFor
\EndFor
\State return true
\EndProcedure
\end{algorithmic}
\end{algorithm}

The full single-threaded procedure is described in Algorithm
\ref{alg:single}.
For every layer, the procedure performs the following actions:
(1) Verify that the node is valid with respect to the justification claimed using the node's and its parents' syntactic information and the
parents' assumption information.
(2) If valid, update the \texttt{assumptions} data structure for the current node based on the parents' assumptions and justification.

In order to better highlight the progression of the algorithm, we will
walk through an example by looking at the
verification of Figure \ref{fig:NDtoHypergraphs} (Right).
Subscripts for the propositions help to distinguish between 
formulae by referencing the ID denoted inside the purple box in the figure.
In the beginning of the algorithm, the first layer only contains assumptions:
\begin{equation*}
    \texttt{currentLayer} = \{(\neg A)_1, (\neg B)_2, A_5, B_7, (A \vee B)_8\}
\end{equation*}
We then go through each node and verify them. Since
they are justified as assumptions, they are trivially valid.
The nodes then have their assumptions updated. The next layer only
contains $(\neg A \wedge \neg B)_3$. This validates and the node's assumptions are
updated to $\{(\neg A)_1, (\neg B)_2\}$. The third layer only contains $(\neg A)_4$.
This validates and the assumptions are propagated forward. On the fourth layer,
the node $B_6$ is justified by negation elimination. This validates and the node's assumptions are set to $\{(\neg A)_1, A_5\}$.
The fifth and final layer only contains the node $B_9$.
As the node is justified
by disjunctive elimination and is valid, we update the assumptions to $\{(\neg A)_1, (A \vee B)_8\}$. As we have gone through all the layers successfully, the whole hypergraph
is valid.

\begin{theorem}
    For all Proof Graphs $p$ the single-threaded \ttfamily \textsc{VERIFY$(p)$} \rmfamily is correct with respect to the validity of the \ND proof corresponding to $p$.  
\end{theorem}

\begin{proof} An algorithm is \textit{correct} if it is sound and complete (termination is trivial). We prove completeness and soundness follows from symmetry.
    A natural deduction proof is valid if all steps are valid. For a step to be valid it needs to pass the \emph{syntactic transformation} and the \emph{assumption constraints}.
    From these, only the assumption constraint check
    requires information from outside the node and its
    parents.
    The rules of natural deduction in
    Figure \ref{fig:rules} show how assumptions
    are computed based on the parent node's assumption
    sets.
    As such, parents of a node must be verified
    beforehand and have their assumptions computed.
    Nodes that are justified via assumptions mark the
    base case of this procedure as their assumption
    set only contains itself.
    Due to the definition of the layering in Equation \ref{eqn:layer} and its usage in Line 3, assumptions
    are in the first layer and the parents of a node
    must be in the previous layer.
    This means that the parents are verified and their assumptions are computed beforehand on lines 6-9.
    Inductively this means that all nodes are verified and have their assumptions computed successfully.
    Hence, the hypergraph proof itself is verified.
\end{proof}

\begin{figure}[ht]
    \centering
    \begin{minipage}{0.7\textwidth}
        \centering
        \includegraphics[width=\textwidth]{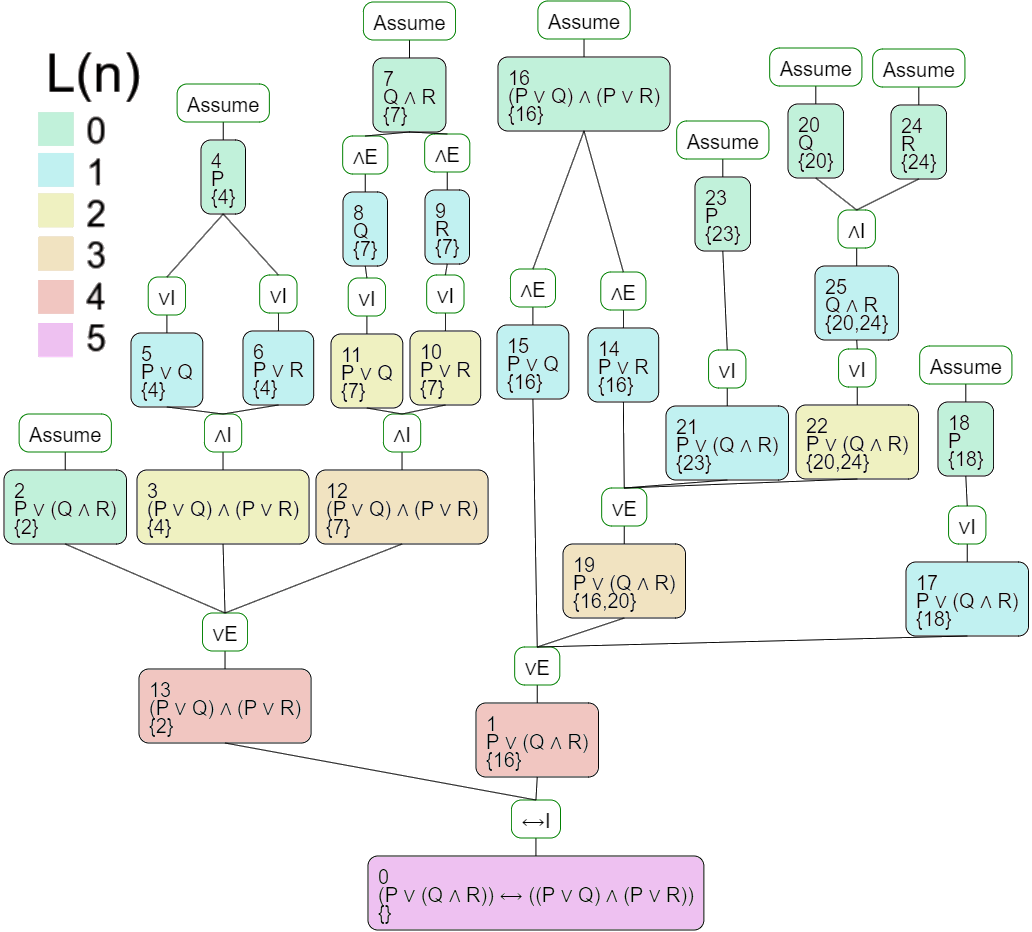}
        
    \end{minipage}\hfill
    \begin{minipage}{0.2\textwidth}
		\begin{tabular}{ll}
			\toprule
			Layer & Nodes\\
			\midrule
			0  & 2,4,7,16, \\ 
               & 23,20,24,18 \\
			1  & 5,6,8,9,15, \\
               & 14,21,25,17  \\
			2  & 3,11,10, 22 \\ 
			3  & 12,19 \\
			4  & 1, 13 \\
			5  & 0   \\
			\bottomrule
		\end{tabular}
    \end{minipage}
    \caption{(Left) A proof of logical or ($\lor$) distributivity over logical and ($\land$). (Right) The nodes of the left proof grouped by layer.}
    \label{fig:example}
\end{figure}

\subsection{Parallel Implementation (Non-optimized Parallel)}

Within a layer, each node only depends on nodes in
layers prior. This means that a node on some layer $n$
does not depend on any other node on layer $n$.
For our initial parallel implementation, fully described
in Algorithm \ref{alg:parallel}, we take advantage of this
and verify the validity of each node on the same
layer in parallel.
In our implementation, we used a shared memory approach.
To combat thread-safety issues we introduce a
vector for each given layer called \texttt{aIds}.
This vector lives in shared memory. Each entry
holds a set of node ids and its length is
the number of nodes in the current layer.
After computing the layers, each node gets evenly distributed
to the available threads.
Each individual thread would then verify the nodes it
was assigned and update the assumptions of the appropriate
index within \texttt{aIds}.
This is thread-safe because each entry only gets written to by the
thread that its corresponding node id has been assigned to.
The parallel portion of the algorithm additionally
performs an \texttt{AND}
reduction on the verification result of each node.
This means that if any result of the individual verifications
is false, then the entire verification result is false.
After the parallel portion is finished, the global
\texttt{assumption} map gets updated using \texttt{aIds}.

Let us illustrate the distribution of nodes with
the example from Figure \ref{fig:example}.
In the graphic, the number at the top of each
node represents its identifier. For brevity,
we will use those numbers when referring to
the nodes.
As the zeroth layer only contains assumptions that
are trivially verifiable, we will start our discussion
from layer one.
In this layer, we have the following nodes:
\begin{equation*}
    L(1) = \{5, 6, 8, 9, 15, 14, 21, 25, 17\}
\end{equation*}
For the sake of example, let us say we have
three threads.
Then thread 0 would be assigned 
$\{5, 6, 8\}$, thread 1 would be assigned
$\{9, 15, 14\}$, and thread 2 would be
assigned $\{21, 25, 17\}$.
Let us focus on thread 0.
Each of the three nodes verifies and the following
assignments are made to \texttt{aIds} based on
the justification used:
\begin{align*}
    aIds(5) &= \{assumptions(parent(5))\} \\
    aIds(6) &= \{assumptions(parent(6))\} \\
    aIds(8) &= \{assumptions(parent(8))\}
\end{align*}
After the parallel portion, the data within
\texttt{aIds} are copied into \texttt{assumptions}
as a way to ensure thread-safety.
On the next layer, we have the following
nodes: $L(2) = \{3, 11, 10, 22\}$.
This then gets distributed with
thread 0 getting $\{3, 11\}$, 
thread 1 obtaining $\{10\}$,
and thread 2 obtaining $\{22\}$.
Focusing on the first thread again, the
nodes verify and the following assignments
are made to \texttt{aIds}:
\begin{align*}
    aIds(3) &= \{assumptions(5), assumptions(6)\} \\
    aIds(11) &= \{assumptions(8)\}
\end{align*}
Notice that each of the items in those sets contains
assumptions that were computed in the previous layer.
They were originally assigned within \texttt{aIds} but 
then copied to \texttt{assumptions}.
The results of each individual node verification
were stored in \texttt{NodeValid} which is then
AND-reduced into the variable \texttt{LayerValid}.
Hence after the end of the parallel portion,
the variable \texttt{LayerValid} would be true
unless one of the nodes in the parallel portion
failed to verify. 

\begin{algorithm}[!hpt]
\caption{Multi-Threaded Algorithm}
\label{alg:parallel}
\begin{algorithmic}[1]
\Procedure{verify}{ProofGraph p}
\State Initialize \texttt{assumptions} to be empty.
\State Create set of nodes on each layer using Equation \ref{eqn:layer} and store in layerMap.

\For {layerNodes in layerMap}
    \State nl = length(layerNodes)
    \State \texttt{aIds} = sharedVector(length=nl)
    \State LayerValid = true
    \For {n in layerNodes in \emph{parallel}}
        \State justification = \texttt{just}(n)
        \State ruleInfo = (m, \texttt{assumptions}(m)) $\forall$ m $\in$ parents(n)
        \State NodeValid = is\_valid(n, justification, ruleInfo)
        \If {NodeValid}
            \State Update aIds(n) using justification and ruleInfo
        \EndIf
        \State AND\_reduce(LayerValid, NodeValid)
    \EndFor
    \If {not LayerValid}
        \State return false
    \EndIf
    \State Update assumptions(n) using aIds
\EndFor
\State return true
\EndProcedure
\end{algorithmic}
\end{algorithm}

\subsection{Multi-Threaded Static Load Balancing Optimization (Load Balancing)}

Notice when going over Figure \ref{fig:example} in the last
example that the distribution of nodes in the second layer was uneven.
The assignment had one thread verifying two while the others
verifying only a single one.
In fact, we can speak to this more generally.
Let $T$ represent the number of threads available and $l_i$ be the
number of nodes in layer $i$.
Let us assume that the verification of one node takes two units of work total:
one to verify the syntax and one to verify the assumptions.
Let $m = l_i \bmod T$. Then, if $m \neq 0$ there are $(T - m)$ threads that are doing
one less unit of work. 

\begin{algorithm}[ht]
\caption{Multi-Threaded Load Balance Algorithm}
\label{alg:load-balance}
\begin{algorithmic}[1]
\Procedure{verify}{ProofGraph p}
\State Initialize \texttt{assumptions} to be empty.
\State Create set of nodes on each layer using Equation \ref{eqn:layer} and store in layerMap.
\State AllNodes = Flatten(layers)

\State numSyntaxVerified = 0
\For {layerNodes in layerMap}
    \State nl = length(layerNodes)
    \State \texttt{aIds} = sharedVector(length=nl)
    \State LayerValid = true
    \State layerSyntaxVerified = 0
    \For {n in layerNodes in \emph{parallel}}
        \State justification = \texttt{just}(n)
        \State ruleInfo = (m, assumptions(m)) $\forall$ m $\in$ parents(n)
        \State NodeValid = is\_valid(n, ruleInfo)
        \If {NodeValid}
            \State Update aIds(n) using justification and ruleInfo
        \EndIf
        \State threadIterSyntaxVerified = 1
        \If {thread verifying less nodes}
            \State extraN = AllNodes[numSyntaxVerified + threadId]
            \State NodeValid = NodeValid and syntaxVerify(extraN, parents(extraN))
            \State threadIterSyntaxVerified = 2
        \EndIf
        \State SUM\_reduce(layerSyntaxVerified, threadIterSyntaxVerified)
        \State AND\_reduce(LayerValid, NodeValid)
    \EndFor
    \If {not LayerValid}
        \State return false
    \EndIf
    \State Update assumptions(n) using aIds
    \State numSyntaxVerified = numSyntaxVerified + layerSyntaxVerified
\EndFor
\State return true
\EndProcedure
\end{algorithmic}
\end{algorithm}

It is with this consideration that we look at static load balancing, presented in Algorithm \ref{alg:load-balance}.
For the threads with one less unit of work,
they take a node from a future layer and syntax verify them.
This approach is valid because syntax verification only requires
the current node and its parents' formulae which are stored
in the proof graph and does not require additional
information from the prior layers such as assumption sets. 
In order to ensure that the nodes that are syntax verified
by the remaining $(T - m)$ threads are distinct, we make use of
another reduction variable \texttt{numSyntaxVerified}. 
Each thread would be assigned
the node that's the sum of that variable and its thread id.
Do note that this is different from dynamic load balancing as
the amount of work is evenly distributed and does not take into
account during runtime some threads finishing before others.

Let us turn to our example from
Figure \ref{fig:example} again.
Recall that the distribution of work
at layer two was the following:
thread 0 maps to $\{3, 11\}$, 
thread 1 has $\{10\}$,
and thread 2 has $\{22\}$.
We can then squeeze in syntax
verification checks in thread 1 and thread 2.
Then, the new allocation becomes:
thread 0 maps to $\{3, 11\}$,
thread 1 maps to $\{10, x_s\}$,
and thread 2 maps to $\{22, y_s\}$.
The subscript denotes how we are only
performing a syntax verification at that
step. Recall that we can not perform
full verification of nodes in future layers
because we do not know if there's a node on the
current layer that its assumptions depends on.
The question then is: how are $x_s$ and $y_s$
calculated? 
As noted before, this is where we keep track of the
total number of nodes that we have syntax verified
already. If we have a flat vector of all nodes
that are partially ordered by their layer number,
then for thread $i$ we can have it syntax verify
$\texttt{numSyntaxVerified} + i$ element of that flat vector.
A sum reduction then keeps track of the total number of syntax verifications
performed on a given layer which is later used to update \texttt{numSyntaxVerified}.

\subsection{Parallel Distribution of Syntax Checks (Syntax-First)}

In the last section we discussed that syntax verification can occur
beyond the current layer being considered. In fact, syntax verification
can occur outside of the layering structure in general which this optimization
considers. In this approach we first perform the syntax verification in parallel
over all nodes before iterating over the layers. This approach is outlined in Algorithm \ref{alg:syntax-first}.
This not only has the benefit of lowering the time to find a syntactic error,
but also more evenly distributes the syntax verification over all threads.

For our example in Figure \ref{fig:example}, the proof graph contains
node ids $0$ through $25$.
If we have three threads, then thread $0$ would be assigned
nodes $0$ through $8$, thread $1$ would be assigned nodes $9$ through $16$,
and thread $2$ would be assigned nodes $17$ through $25$.
Each thread would then loop over their assigned nodes and syntax verify
them. When the three threads finish, if any of their nodes failed to syntax verify,
then the algorithm would end and the proof graph verification would fail.
In this example, however, the nodes pass the syntax verification.
The rest of the algorithm closely follows Algorithm \ref{alg:parallel}
where instead of performing a full verification, we only verify that the
assumption constraints hold.

\begin{algorithm}[ht]
\caption{Multi-Threaded Syntax Check First Algorithm}
\label{alg:syntax-first}
\begin{algorithmic}[1]
\Procedure{verify}{ProofGraph p}
\State syntaxValid = True
\For {n in p.nodes in \emph{parallel}}
    \State valid = verifySyntax(n, parents(n))
    \State AND_reduce(syntaxValid, valid)
\EndFor

\If {not syntaxValid}
    \State return false
\EndIf

\State Initialize \texttt{assumptions} to be empty.
\State Create set of nodes on each layer using Equation \ref{eqn:layer} and store in layerMap.

\For {layerNodes in layerMap}
    \State nl = length(layerNodes)
    \State \texttt{aIds} = sharedVector(length=nl)
    \State LayerValid = true
    
    \For {n in layerNodes in \emph{parallel}}
        \State justification = just(n)
        \State ruleInfo = (m, assumptions(m)) $\forall$ m $\in$ parents(n)
        \State NodeValid = verifyAssumptions(n, justification, ruleInfo)
        \If {NodeValid}
            \State Update aIds(n) using justification and ruleInfo
        \EndIf
        \State AND\_reduce(LayerValid, NodeValid)
    \EndFor
    \If {not LayerValid}
        \State return false
    \EndIf
    \State Update assumptions(n) using aids
\EndFor
\State return true
\EndProcedure
\end{algorithmic}
\end{algorithm}

\section{Methodology and Results}\label{sec:Results}
To discuss the performance of our algorithms, we provide
an empirical investigation.
To this end, we look toward a comparison
of the number of seconds needed to verify various proof
structures using the algorithms described before.
Inspired by the topologies used in computer
network design \cite{bicsi2002network}, 
we introduce a directed variant that we call
directed acyclic network topologies or DANTs.
These DANTs represent different classes of possible
proofs with which we perform benchmarks over.

\subsection{Directed Acyclic Network Topologies (DANTs)}
\label{sec:DANT}

\begin{figure}
    \begin{minipage}{.5\textwidth}
        \centering
        \includegraphics[width=7cm]{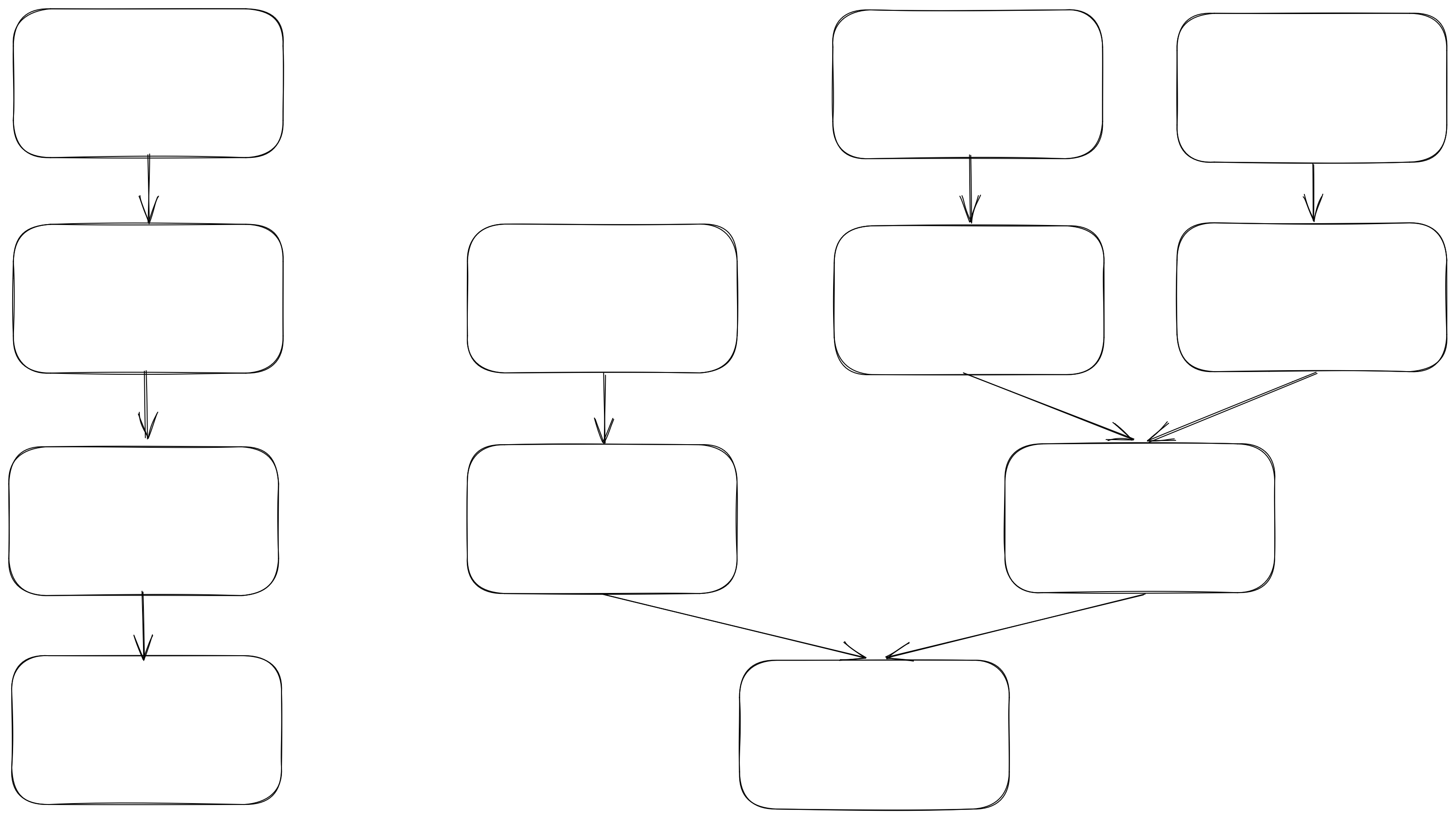}
    \end{minipage}
    \begin{minipage}{.5\textwidth}
        \centering
        \includegraphics[width=7cm]{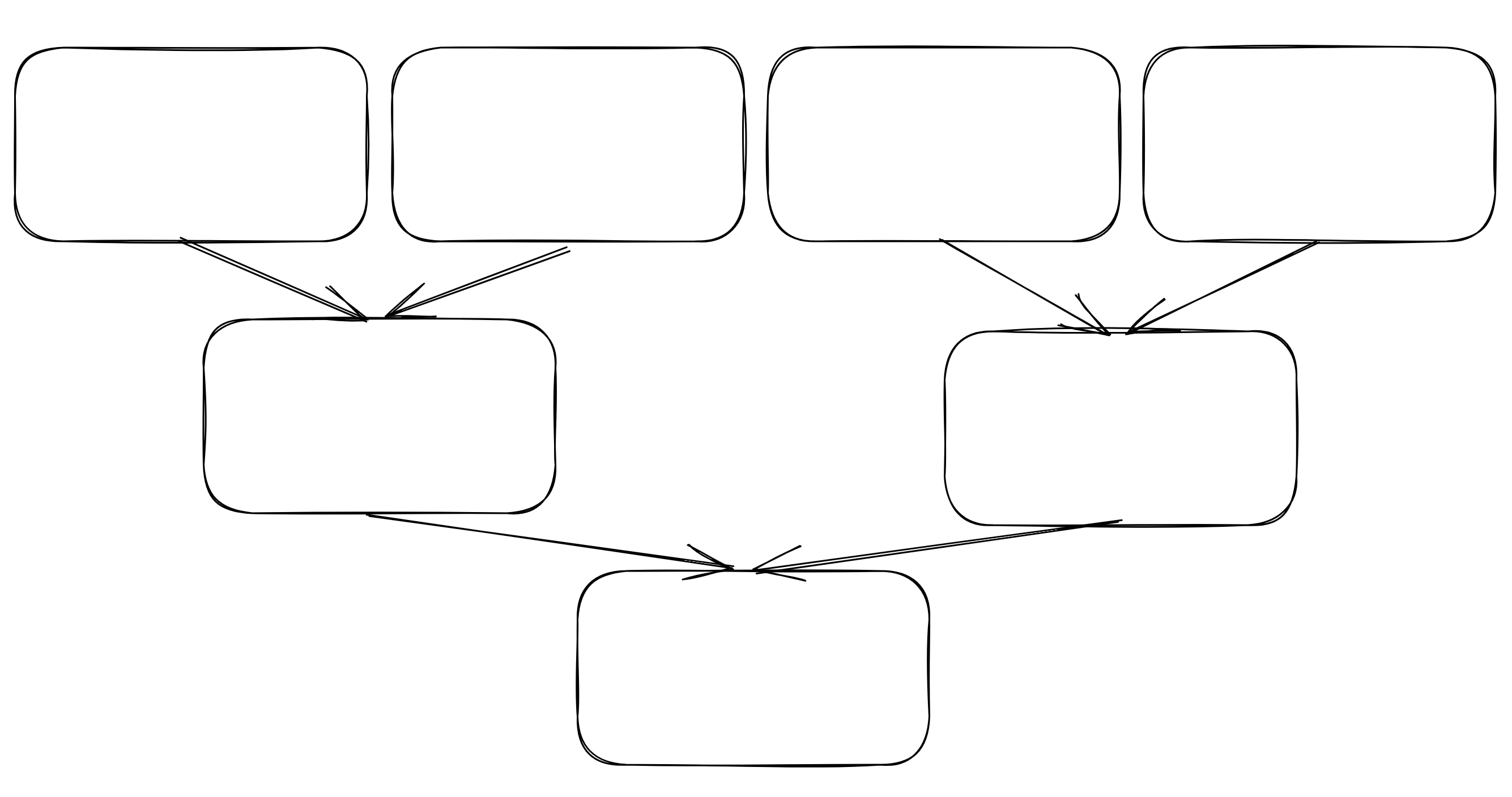}
    \end{minipage}
    
    \caption{Straight, Parallel-Branch, then Tree Topologies}
    \label{fig:topologies}
\end{figure}

We have identified three distinct classes of DANTs 
which we include in our benchmark of proof graphs for
analysis. The pictorial depiction is shown in
Figure \ref{fig:topologies} and  is generated as so:
    \textbf{Straight Line $(n)$}: Parameterized by the total
    number of nodes $n$, this topology only contains one assumption
    at the top and then each future inference is a disjunction
    introduction. This enforces a straight linear proof with no
    branching. We do not expect any speedup in the parallel algorithm
    in this case as each layer only has one node.  
    \textbf{Parallel Branches $(b, n)$}: This topology emulates multiple
    lines of independent reasoning before combining towards the end.
    It starts off with $b$ separate assumptions and then performs a
    disjunctive introduction on each assumption $n$ times before
    iteratively applying conjunctive introduction to each branch
    until there is one remaining. 
    We expect the number of branches will correlate with the scalability of verification. It should be noted that this topology is isomorphic to a straight topology of length $n$ when $b=1$ and emulates a tree like topology of height $b$ when $b > 1$.
    \textbf{Tree $(h)$}: In this topology we generate $2^h$ assumptions and iteratively
    apply conjunction introduction $h$ times until we reach a single node.
    This creates a balanced binary tree.
    We hypothesize the greatest amount of speedup from this topology.

\subsection{Empirical Analysis}

We perform two classes of experiments: (1) a parallel strong scaling study in which the proof to be verified is held constant while the number of processors increases; (2) a problem size scaling study in which we hold the number of threads constant and look at how each method performs as the problem gets harder.

\subsubsection{Implementation Details}
Our benchmarks were performed on one node of the IBM DCS supercomputer, AiMOS, at Rensselaer Polytechnic Institute. Our code is available at \url{https://github.com/RAIRLab/Parallel-Verifier}
\footnote{For reproducibility of our results, please see the following link for the specific commit
the results of this paper is based on: \url{https://github.com/RAIRLab/Parallel-Verifier/tree/a661abbe5bf038a3fa8645b8af532b0a60daebe5}}. 
The code is implemented in C++11 and makes use of the standard C++ library data structures.
For the multiprocessing component, we use
the OpenMP library \cite{chandra2001parallel}. OpenMP operates over software threads which are assigned to CPUs.
We ensure during our scaling study that the system is not \textit{oversubscribed}, meaning that there is just a single thread used per CPU. AiMOS provides us a single node on which ten physical cores are available; however one is reserved for the operating system and IO, therefore nine are used for our experiments.
For our benchmarks we do not include the time it took for initialization, file parsing, or proof parsing; 
we only measure the time taken to verify the proof.
For this, we record the clock-cycles before and after the execution of the verification algorithm and use their difference to compute the total cycles.
We then compute the number of seconds taken by each method through dividing the number of total cycles by the base clock rate of 512MHz.

\subsubsection{Strong Scaling}\label{sec:strongscaling}
For our strong scaling study we vary the number of threads used while holding the DANT instance constant. 
For the straight topology, we consider a length $n$ of 150.
For the branch topology, we consider $b = 150$ branches each with a length of $n = 100$.
Lastly, for the tree topology, we consider $h = 16$ conjunction introductions for a
total number of $2^{16}$ vertices.
Results can be seen on the left of Figure \ref{fig:results}. The strong scaling results show a clear benefit to our parallelized verification approach. In the case of the straight topology the serial algorithm vastly outperforms the parallel algorithms, which is expected as in this topology
there is only one node per layer.
There are clear overheads to parallelization,
such as waiting for all threads to finish, 
that make timing differences visible as the number of threads increases for the straight topology.
For the branch topology with 150 branches we see that our parallel methods scale well, particularly load balancing which beats out syntax-first and non-optimized parallel methods. We hypothesize this is due to the number of remaining nodes on each layer remaining constant which allows for a good balance of syntax checking vs assumption updating. The parallel methods perform quite well on the tree topology significantly beating out the serial method, with non-optimized parallel and syntax-first methods beating out load balancing likely due to the overhead costs.

\subsubsection{Problem Scaling}
For our problem scaling study we hold the number of threads constant (at AiMOS' maximum value of nine) and vary the problem size. For the straight topology, we consider a chain of disjunctive introductions of lengths ($n$) 100, 150, 200, 250, 300, 350, and 400.  For the branch topology, we consider a fixed branch length of $n=100$ and vary the number of branches ($b$) at 30, 50, 70, 90, 110, 130, and 150. Lastly for the tree topology, we create binary trees of heights ($h$) 8, 10, 12, 14, 16, 18, and 20.
Results can be seen on the right side of Figure \ref{fig:results}. We hypothesize the straight topology scaling is not linear due to overheads such as the formulae length increasing as the problem size increases. We see that for all problem sizes on the straight topology, the serial implementation outperforms the parallel implementation. This aligns with the observation in the strong scaling study that the parallel methods have overheads and the fact that for all parallel methods, only one node is on each layer, preventing the majority of the threads from doing any work. In the branch topology, the results show that as the number of branches increases, the effectiveness of parallel methods increases.
This is particularly shown in load balancing, due to the reasons discussed in \S \ref{sec:strongscaling}.
For the tree topology, there is an exponential increase in the time taken as the problem grows, largely due to the fact that the number of nodes to verify increases exponentially ($2^n$) as the problem size increases. We see that as the problem size grows, the performance of our parallel methods over the serial method increases substantially. 

\begin{figure}
    \centering
    \includegraphics[width=0.49\textwidth]{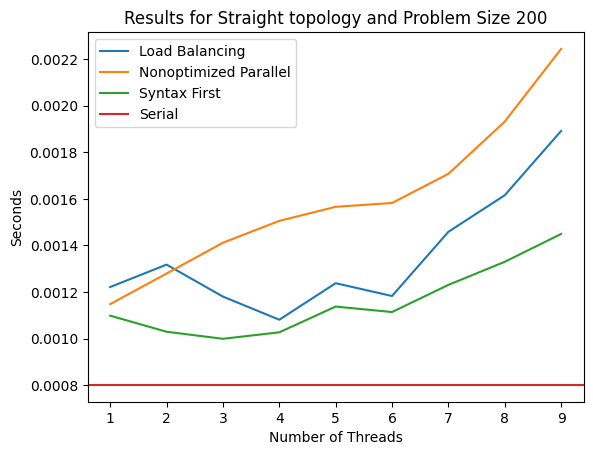}
    \includegraphics[width=0.49\textwidth]{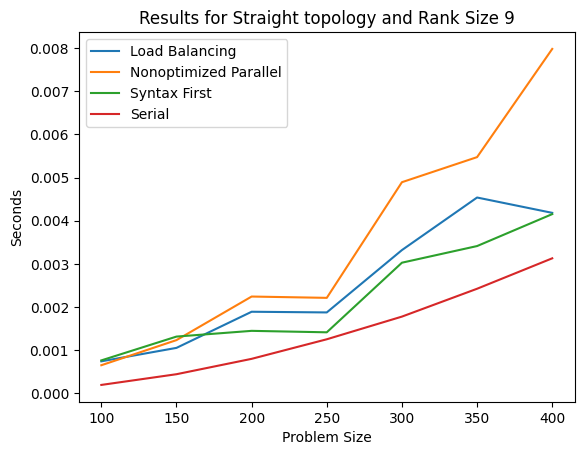} \\
    \includegraphics[width=0.49\textwidth]{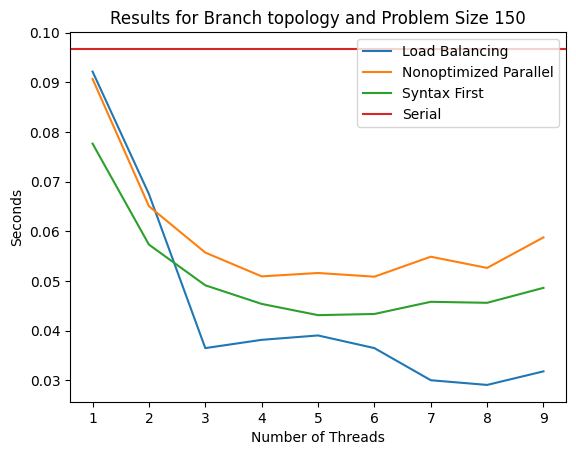}
    \includegraphics[width=0.49\textwidth]{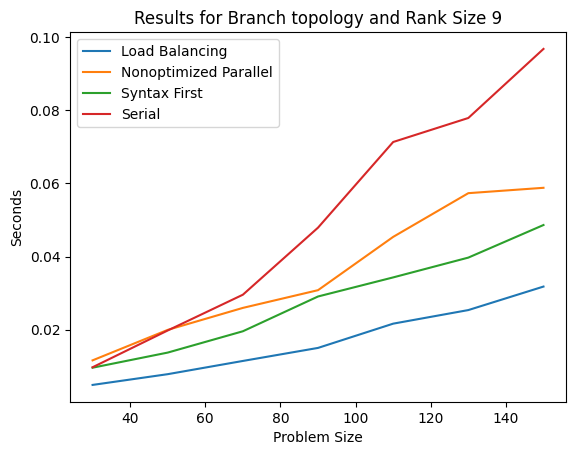}\\
    \includegraphics[width=0.49\textwidth]{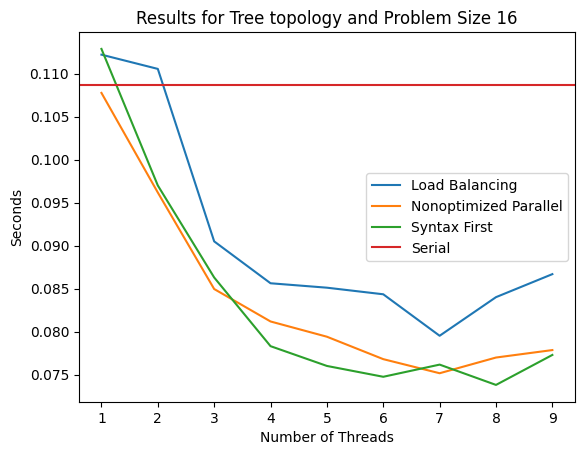}
    \includegraphics[width=0.49\textwidth]{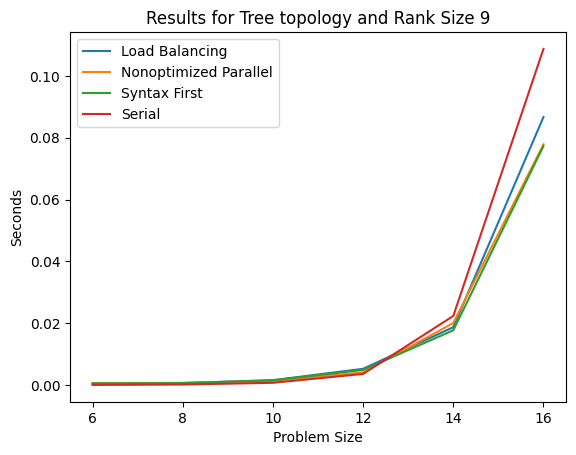}
    \caption{(Left) Strong scaling study results (Right) Problem scaling study results.}
    \label{fig:results}
\end{figure}

\section{Related Work}\label{sec:Related-Work}
Past work has investigated parallel or concurrent
verification of other logical calculi.
For example the developers of Isabelle
\cite{wenzel2008isabelle}, a proof assistant
with support for first-order logic,
higher-order logic, and Zermelo-Fraenkel set theory,
used concurrent programming for the efficient
verification of LCF-style proofs \cite{scott1993type}.
In their work, during the verification of a proof,
if a reference to another proof is made, and
that proof has yet to be verified, then 
a promise is deferred.
At the end of verification, all promises are resolved
and localized errors are then shown to the user
\cite{parallelisabelle, wenzel2013read}.
This work operates over entire proofs
while we focus on parallelizing steps
within a proof.
In \cite{Wenzel}, Wenzel introduces what he calls
granularities of concurrency within verifying
a single proof. The levels
include concurrent verification of theories,
concurrent verification of commands, and
concurrent verification of subproofs.
In terms of our work, we do not consider
extra background theories, commands correspond
to our proof steps, and we do not consider
subproofs in our work. As discussed
in \S \ref{sec:Background}, there is a definition of
a subproof in our natural deduction hypergraphs.
However, it's not something specified by the creator
of the proof and there
can be $n$ different subproofs for a proof with $n$ nodes.

F{\"a}rber looked at concurrent verification of commands in his work
parallelizing proof checking inside the lambda-Pi calculus
modulo rewriting \cite{Frber2022SafeFC}. 
In this work, he breaks up a command into four tasks:
parsing, sharing, type inference, and type checking. 
Similarly, our work breaks up our inference rules into two steps:
syntactic checks and assumption checks. We additionally look
at the parallel verification of sets of steps or commands, as opposed
to only looking at the concurrency within each command.

\section{Conclusion}
\ignore{We hope that the introduction of these DANTs inspires discussion
on a null model of proof graphs or a benchmark of graphical style
proofs which can be used in the performance analysis of verification
algorithms.}

In this work, we presented a layering based algorithm that decouples
the underlying semantic dependencies of proof steps in
natural deduction.
Through this, we introduced a suite of new algorithms
which use layers to parallelize verification
of hypergraphical natural deduction proofs.
Directed acyclic network topologies (DANTs) were introduced
as a benchmark for hypergraphical proofs and we have
shown in our analysis that the parallel algorithms
perform better than their serial counterpart
on non-straight DANT instances.
These parallel algorithms were additionally shown
to scale through both
the strong scaling and problem scaling studies. 
This work has applications in formal verification, specifically in
proof assistants.

Our future work falls into four categories: theoretical results, empirical results, logic extensions, algorithmic optimizations: (1) For theoretical results, we would benefit from analysis with respect to Amdahl's Law \cite{Hill2008}
to calculate the overall speedup with respect to different parallelizable tasks in each of the algorithms.
(2) For further empirical results, we can test randomized proof topologies or craft a dataset of
common natural deduction proofs. 
(3) We wish to extend our verifier to handle different types of logics,
    specifically first-order and modal logics. First-order logics are used
    heavily within proof assistants, and require the ability to represent and reason over formulae at the term level, including the need for checking if variables are free or bound in inference rules. We conjecture that despite these extra requirements, our layer based parallel approaches would still work on first-order proof graphs. In order to handle modal logics such as \textbf{K5}, we would have to adapt the algorithm to include additional
    bookkeeping required for several of the inference rules.
(4) We would like to explore approaches to scale beyond a single computer. This involves exploring message-passing parallelism which is often used in distributed computing. 

\paragraph{Acknowledgements}
This paper was supported in part by a fellowship award under contract FA9550-21-F-0003 through the National Defense Science and Engineering Graduate (NDSEG) Fellowship Program, sponsored by the Air Force Research Laboratory (AFRL), the Office of Naval Research (ONR) and the Army Research Office (ARO).

\bibliographystyle{eptcs}
\bibliography{references}

\end{document}